\documentclass[a4paper,USenglish,cleveref, autoref, thm-restate, numberwithinsect]{lipics-v2021}

\newif\iflong
\newif\ifshort

\longtrue

\iflong
\else
\shorttrue
\fi
\iflong
\hideLIPIcs  %
\fi

\title{On Minimizing Wiggle in Stacked Area Charts} 
\titlerunning{On Minimizing Wiggle in Stacked Area Charts}

\author{Alexander Dobler}{TU Wien, Austria}{adobler@ac.tuwien.ac.at}{0000-0002-0712-9726}{Supported by the Vienna Science and Technology Fund (WWTF) under grant [10.47379/ICT19035]}

\author{Martin Nöllenburg}{TU Wien, Austria}{noellenburg@ac.tuwien.ac.at}{0000-0003-0454-3937}{Supported by the Vienna Science and Technology Fund (WWTF) under grant [10.47379/ICT19035]}

\authorrunning{A.\ Dobler and M.\ Nöllenburg}

\Copyright{Jane Open Access and Joan R. Public}



\relatedversion{} %

\supplement{Source code is available online \cite{gitzenodo}.}%

    \nolinenumbers %

\EventEditors{John Q. Open and Joan R. Access}
\EventNoEds{2}
\EventLongTitle{42nd Conference on Very Important Topics (CVIT 2016)}
\EventShortTitle{CVIT 2016}
\EventAcronym{CVIT}
\EventYear{2016}
\EventDate{December 24--27, 2016}
\EventLocation{Little Whinging, United Kingdom}
\EventLogo{}
\SeriesVolume{42}
\ArticleNo{23}

\usepackage{graphicx} %

\usepackage[basic]{complexity}

\usepackage{amsmath,amsfonts,amsthm,amssymb}
\usepackage{thmtools}
\usepackage{hyperref}
\usepackage[ruled,linesnumbered,vlined]{algorithm2e}
\usepackage{mdframed}
\usepackage{nccmath}
\crefname{algocf}{Algorithm}{Algorithms}
\usepackage{todonotes}
\usepackage{printlen}
\usepackage{cite}
\usepackage{apptools}
\usepackage{xpatch}

\usepackage{tabularx, environ}
\newcommand{\probname}[1]{{\normalfont\textsc{#1}}}

\newcommand{\pwigglemin}{\probname{$p$-WiggleMin}}
\newcommand{\wpwigglemin}{\probname{Weighted-$p$-WiggleMin}}
\newcommand{\onewigglemin}{\probname{$1$-WiggleMin}}
\newcommand{\wonewigglemin}{\probname{Weighted-$1$-WiggleMin}}
\newcommand{\prefsumsumprob}{\probname{Min-$\sum$\textbar Prefixsum\textbar}}
\newcommand{\oneintwopartition}{\probname{OneInTwoPartition}}
\newcommand{\minlinarr}{\probname{MLA}}
\newcommand{\minlinarrlong}{\probname{Minimum Linear Arrangement}}

\newcommand{\mylipgray}[1]{\textcolor{lipicsGray}{\sffamily\bfseries\upshape\mathversion{bold}#1}}

\newcommand{\MILP}{\mylipgray{WiggleMILP}}
\newcommand{\Upwards}{\mylipgray{UpwardsOpt}}
\newcommand{\bestfirst}{\mylipgray{BestFirst}}

\theoremstyle{definition}
\newtheorem{problem}{Problem}

\usepackage{etoolbox}
\newcommand{\appsymb}{$\bigstar$}
\newcommand{\appref}[1]{\hyperref[proof:#1]{\appsymb}}

\newcommand{\toappendix}[1]{%
  \iflong{}#1\else{}
    \gappto{\appendixText}
    {
        #1
      }
  \fi{}%
}

\newcommand{\appendixproof}[2]{%
  \iflong{}#2\else{}\gappto{\appendixText}
    {
      \subsection{\texorpdfstring{Proof of \cref{#1}}{}}\label{proof:#1}
      #2
    }
  \fi{}
}

\iflong
\makeatletter
\xpatchcmd{\thmt@restatable}%
{\csname #2\@xa\endcsname\ifx\@nx#1\@nx\else[{#1}]\fi}%
{\csname #2\@xa\endcsname}%
{}{} %
\makeatother
\fi

\begin{document}

\maketitle
\begin{abstract}
    Stacked area charts are a widely used visualization technique for numerical time series. The $x$-axis represents time, and the time series are displayed as horizontal, variable-height layers stacked on top of each other. The height of each layer corresponds to the time series values at each time point. The main aesthetic criterion for optimizing the readability of stacked area charts is the amount of vertical change of the borders between the time series in the visualization, called \emph{wiggle}. While many heuristic algorithms have been developed to minimize wiggle, the computational complexity of minimizing wiggle has not been formally analyzed. In this paper, we show that different variants of wiggle minimization are \NP-hard and even hard to approximate. We also present an exact mixed-integer linear programming formulation and compare its performance with a state-of-the-art heuristic in an experimental evaluation.
    Lastly, we consider a special case of wiggle minimization that corresponds to the fundamentally interesting and natural problem of ordering a set of numbers as to minimize their sum of absolute prefix sums. We show several complexity results for this problem that imply some of the mentioned hardness results for wiggle minimization.
\end{abstract}
\section{Introduction}
\begin{figure}[tb]
    \centering
    \includegraphics{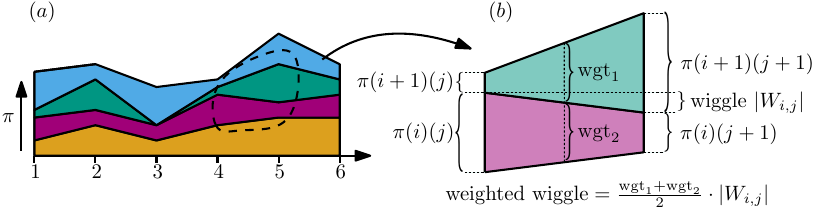}
    \caption{(a) A stacked area chart of four time series on six time points. The time series are ordered according to $\pi$ from bottom to top. (b) Illustration of the wiggle between time points $4$ and $5$ on the border between the green and violet time series. For the weighted wiggle, $|W_{i,j}|$ is multiplied by the average of $\pi(i)(j)$, $\pi(i)(j+1)$, $\pi(i+1)(j)$, and $\pi(i+1)(j+1)$. The total (weighted) wiggle is the sum of all (weighted) wiggle values.}
    \label{fig:stackedareachart}
\end{figure}
Stacked area charts are a widely used method for visualizing numerical time series across discrete time points, such as population data of countries over time or movie revenues over multiple weeks \cite{huggett_multiple_1990,byronStackedGraphsGeometry2008}. In stacked area charts, the $x$-dimension depicts time, and the time series are stacked as variable-height strips on top of each other without gaps, such that their heights depict their values (see \cref{fig:stackedareachart}a). The lowest border of the stacked time series is a straight horizontal line. 

A primary aesthetic criterion for stacked area charts is \emph{wiggle} -- the aggregated amount of vertical change of the borders between the time series in the visualization (\cref{fig:stackedareachart}b). Further, the wiggle is usually weighted by the time series values, resulting in \emph{weighted wiggle} (for a formal definition refer to \cref{section:prelims}).
For a stacked area chart of a specific dataset, the amount of wiggle solely depends on the vertical ordering of the time series. The task of minimizing wiggle has already been tackled with heuristic methods \cite{byronStackedGraphsGeometry2008,strungemathiesenAestheticsOrderingStacked2021}. But despite the popularity of stacked area charts in mainstream visualizations, no work was done with regard to the computational complexity of minimizing wiggle. This paper fills the gap and presents several complexity lower bounds for minimizing wiggle and weighted wiggle. We also compare a state-of-the-art heuristic for minimizing weighted wiggle with a new mixed-integer linear programming formulation on a set of real-world data; the evaluation shows that the heuristic performs well with regard to wiggle minimization and that the mixed-integer program can solve small to medium-size instances exactly.
Further, a special case of minimizing wiggle is discussed that results in a fundamentally interesting problem -- ordering a set of numbers such that the accumulation of all absolute prefix sums is minimized. Despite its natural problem definition, the problem has not been studied yet.
Its relevance comes from its relation to several ordering problems that minimize accumulated cost, such as scheduling problems. Results for this problem are also used to show some hardness results for minimizing wiggles.

\subparagraph*{Related work.}
The origin of stacked area charts is unclear, as they are a very natural and frequently used visualization for a commonly occurring type of data.
The most notable work for popularizing research on stacked area charts is from 2008 by Byron and Wattenberg \cite{byronStackedGraphsGeometry2008}, since receiving more than 600 citations. Given the extensive literature on stacked area charts, we focus on studies specifically addressing wiggle minimization and similar aesthetic criteria in related visualizations.
Byron and Wattenberg \cite{byronStackedGraphsGeometry2008} only deal with data that represent movie revenues. Such data have very specific properties -- a movie starts airing, its revenues grow quickly, and then the revenues slowly decrease. Thus, they present a very specific ordering approach that orders the movies by their first screening. They also consider a similar type of visualization, called \emph{streamgraphs}. The only difference between streamgraphs and stacked area charts is that the lower border of a streamgraph does not have to be a horizontal line, introducing another degree of freedom. Still, wiggle minimization is defined equivalently for streamgraphs.  Byron and Wattenberg present a similar ordering procedure for streamgraphs.
\iflong Furthermore, they consider the task of computing the lower border of visualization given a fixed vertical order of time series. 
For that task, they show an optimal polynomial time algorithm that minimizes the squared weighted wiggle (each wiggle value is squared instead of taking its absolute value). \fi
Greffard and Kuntz~\cite{greffardVisualizingSetMultiple2015} present an approach for weighted squared wiggle minimization in streamgraphs that first defines pairwise dissimilarity values between time series and then applies traveling salesperson approaches. Di Bartolomeo and Hu~\cite{dibartolomeoThereMoreStreamgraphs2016} present a heuristic and local search algorithm for weighted wiggle minimization in streamgraphs, which also work for stacked area charts. 
\iflong In addition, they  show how to, given a fixed order of time series, compute the lower border of the streamgraph that minimizes weighted wiggle. \fi
Bu et al.~\cite{buSineStreamImprovingReadability2021} present an ordering algorithm for streamgraphs that is based on clustering. They also optimize for minimizing what they call \emph{sine illusion effect} -- time series borders mimicking a sine wave. He and Li~\cite{heOptimalLayoutStacked2022} present an ordering approach for stacked area charts based on a traveling salesperson formulation. Their approach aims to minimize the covariance between adjacent time series in the stacked area chart.
Lastly, Mathiesen and Schulz~\cite{strungemathiesenAestheticsOrderingStacked2021} present a heuristic algorithm that is able to minimize a weighted sum of quality criteria for stacked area charts, including wiggle. Their approach is an improvement over an algorithm by Di Bartolomeo and Hu~\cite{dibartolomeoThereMoreStreamgraphs2016}, and they demonstrate its effectiveness in an experimental evaluation.

The problem of ordering a set of numbers to minimize the accumulation of all absolute prefix sums is related to some classic problems from scheduling. Tsai~\cite{DBLP:journals/orl/Li-Hui92} studies the problem where instead of minimizing the accumulation, one wants to minimize the maximum absolute prefix sum. They claim strong \NP-hardness (a formal proof is absent). Kellerer et al.~\cite{DBLP:journals/ior/KellererKRW98} consider the same problem as Tsai with the additional constraint that each prefix sum must be positive, and show constant-factor approximation algorithms. Further, they explicitly pose the open problem of minimizing the sum of all prefix sums. In their  definition, though, the ordering has to satisfy that each prefix sum is positive, which is not the case for us.

\subparagraph*{Our contribution.}
We consider wiggle minimization and weighted wiggle minimization in stacked area charts from a computational complexity perspective. For this, we define the two computational problems \pwigglemin\ and \wpwigglemin\ in \cref{section:prelims}. The value $p$ in both problem definitions corresponds to the exponent of each wiggle value in the objective function, in turn capturing variants of wiggle minimization such as the weighted squared wiggle minimization from Byron and Wattenberg~\cite{byronStackedGraphsGeometry2008}. Further, we introduce and investigate a special case of instances leading to the problem \prefsumsumprob\ -- ordering a set of numbers to minimize the sum of absolute prefix sums.

In \cref{section:prefsumsum}, we show that \prefsumsumprob\ is strongly \NP-hard and that both wiggle problems with $p=1$ are strongly \NP-hard, even if the number of time points is constant. Next, we also consider special cases of \prefsumsumprob, where there is only one positive or one negative element in the input. \cref{section:pwigglemin} shows that both wiggle problems for arbitrary $p$ are strongly \NP-hard, and hard to approximate under specific complexity assumptions. Further, a lower bound on the approximation ratio of a known greedy heuristic used in \cite{dibartolomeoThereMoreStreamgraphs2016} and \cite{strungemathiesenAestheticsOrderingStacked2021} is shown.
Lastly, \cref{section:experimental} presents a mixed-integer linear program for \wonewigglemin\ and compares it with the state-of-the-art heuristic of Mathiesen and Schulz~\cite{strungemathiesenAestheticsOrderingStacked2021} in an experimental evaluation on real-world data.

\smallskip\noindent
Due to space constraints, statements marked with \appsymb\ are proved in the appendix.

\section{Preliminaries and Problem Definitions}\label{section:prelims}

\iflong\subparagraph*{Permutations.}\fi We define $[n]=\{1,\dots,n\}$. A permutation $\pi$ of a multiset $S=\{s_1,\dots,s_n\}$ is a bijection $\pi:[n]\mapsto S$. We sometimes treat $\pi$ as the list $(\pi(1),\pi(2),\dots,\pi(n))$. Further, we define $\text{pos}_\pi(s)=i$ if and only $\pi^{-1}(s)=i$.

\iflong\subparagraph*{Time series.}\fi An \emph{$\ell$-time series} $f$ is a function $f:[\ell]\to \mathbb{R}^+_0$. We refer to the elements of its image as \emph{data points}.
A set $F$ of $\ell$-time series, is \emph{balanced} if $\sum_{f\in F}f(i)=\sum_{f\in F}f(j)$ for all $i,j\in [\ell]$.

\iflong\subparagraph*{Strong NP-hardness.}\fi A computational problem is \emph{strongly \NP-hard} if it remains \NP-hard even if its numerical parameters are integers that are polynomial in the input size.

\paragraph*{Problem Definitions}
In most problems discussed in this paper, we are given a set $F=\{f_1,\dots,f_n\}$ of $\ell$-time series. Given a permutation $\pi$ of $F$, $i\in \{0\}\cup [n]$, and $j\in [\ell-1]$, we define the \emph{wiggle} value $W^\pi_{i,j}=\sum_{k=1}^i(\pi(k)(j+1)-\pi(k)(j))$. 
This results in the following two problem variants of \emph{wiggle minimization for stacked area charts}, the first unweighted, and the second weighted by the time series data points - as is more common in the stacked area charts literature \cite{dibartolomeoThereMoreStreamgraphs2016,strungemathiesenAestheticsOrderingStacked2021,greffardVisualizingSetMultiple2015} (see also \cref{fig:stackedareachart}b). For both problems, $p$ is a positive integer corresponding to the exponent of the wiggle values. The second problem is equivalent to the minimization of \emph{flatness} in~\cite{strungemathiesenAestheticsOrderingStacked2021}. 
\begin{problem}[\pwigglemin] 
    Given a set $F=\{f_1,\dots,f_n\}$ of $\ell$-time series, find a permutation $\pi$ of $F$ such that $\sum_{i=1}^{n}\sum_{j=1}^{\ell-1}|W^\pi_{i,j}|^p$ is minimized.
\end{problem}

\begin{problem}[\wpwigglemin]
    Given a set $F=\{f_1,\dots,f_n\}$ of $\ell$-time series, find a permutation $\pi$ of $F$ such that 
    \begin{equation*}
        \sum_{i=1}^{n}\sum_{j=1}^{\ell-1}\frac{\pi(i)(j)+\pi(i)(j+1)+\pi(i+1)(j)+\pi(i+1)(j+1))}{4}|W^\pi_{i,j}|^p
    \end{equation*} is minimized. Above, $\pi(n+1)$ is the time series containing only zeroes.
\end{problem}
If $\ell=2$, then \onewigglemin\ is equivalent (see \cref{lemma:prefsumsumtowigglemin}) to the following fundamentally interesting problem, which we also study in this paper.

\begin{problem}[\prefsumsumprob]
    Given a multiset $S=\{s_1,\dots,s_n\}$ of real numbers, find a permutation $\pi$ of $S$ such that $\sum_{i=1}^n \left|\sum_{j=1}^{i}\pi(j)\right|$ is minimized.
\end{problem}
Given a solution $\pi$ of \prefsumsumprob, we define for $i=0,1,\dots,n$, $P^\pi_i=\sum_{j=1}^i\pi(j)$.

Next, we want to show relationships between the previously defined problems.
\begin{restatable}[\appsymb]{lemma}{lemmaprefsumsumtowigglemin}\label{lemma:prefsumsumtowigglemin}
    Given an instance $S=\{s_1,\dots,s_n\}$ of \prefsumsumprob, there exists an instance $F=\{f_1,\dots,f_n\}$ of \onewigglemin\ on two time points such that $S$ has a solution with value $x$ if and only if $F$ has a solution with value $x$. If $\sum_{s\in S}s=0$, then $F$ is balanced. Further, data point values in $F$ are bounded by a polynomial of the values in $S$.
    
    Conversely, for an instance $F$ of \onewigglemin\ on two time points, there exists an instance $S$ of \prefsumsumprob\ such that $S$ has a solution with value $x$ if and only if $F$ has a solution with value $x$.
\end{restatable}
\appendixproof{lemma:prefsumsumtowigglemin}{
\ifshort\lemmaprefsumsumtowigglemin*\fi
\begin{proof}
    Let $S=\{s_1,\dots,s_n\}$ be an instance of \prefsumsumprob. The instance $F=\{f_1,\dots,f_n\}$ of \onewigglemin\ is defined on two time points as follows. For $i\in [n]$, let
    \begin{itemize}
        \item $f_i(1)=0$ and $f_i(2)=s_i$ if $s_i\ge 0$, and
        \item $f_i(1)=-s_i$ and $f_i(2)=0$ otherwise.
    \end{itemize}
    The equivalence between the two instances is immediate since $f_i(2)-f_i(1)=s_i$.

    For the other direction, define $s_i=f_i(2)-f_i(1)$.
\end{proof}
}

\begin{restatable}[\appsymb]{lemma}{lemmawigglemintoweightedwigglemin}\label{lemma:wigglemintoweightedwigglemin}
    Let $p\ge 1$ be an integer.
    Given a balanced instance $F=\{f_1,\dots,f_n\}$ of \pwigglemin\ on $\ell$ time points, there is an instance $F'=\{f'_1,\dots,f'_n\}$ of \wpwigglemin\ on $4\ell+1$ time points and a constant $C$ such that $F$ has a solution with value $x$ if and only if $F'$ has a solution with value $Cx$. Further, the values of the data points in $F'$ are polynomially bounded w.r.t.\ the values of the data points in $F$.
\end{restatable}
\appendixproof{lemma:wigglemintoweightedwigglemin}{
    \ifshort\lemmawigglemintoweightedwigglemin*\fi
    \begin{proof}
    Let $F=\{f_1,\dots,f_n\}$ be a balanced instance of \pwigglemin\ on $\ell$ time points. We construct $F'=\{f'_1,\dots,f'_n\}$ as an instance of \wpwigglemin\ on $4\ell+1$ time points.
    Let $M=\max\{f_i(j)\mid i\in[n],j\in[\ell]\}$.
    We start by defining $f'_i(1)=M$ for all $i\in [n]$. Next, we define inductively for $j=1,\dots,\ell$:
    \begin{itemize}
        \item $f'_i(4j-2)=f'_{i}(4j-3)+(f_i(j+1)-f_i(j))$ for all $i\in [n]$,
        \item $f'_i(4j-1)=f'_{i}(4j-2)-(f_i(j+1)-f_i(j))$ for all $i\in [n]$,
        \item $f'_i(4j)=f'_{i}(4j-1)-(f_i(j+1)-f_i(j))$ for all $i\in [n]$, and
        \item $f'_i(4j+1)=f'_{i}(4j)+(f_i(j+1)-f_i(j))$ for all $i\in [n]$.
    \end{itemize}
    Notice that $f'_i(4j+1)=M$ for all $i\in [n]$ and that the average value of $f'_i(4j-2), f'_i(4j-1), f'_i(4j), f'_i(4j+1)$ is $M$.
    Now consider a permutation $\pi$ of $F$. Consider the permutation $\pi'$ of $F'$ defined such that $\text{pos}_{\pi}(f_i)=\text{pos}_{\pi'}(f'_i)$ for all $i\in [n]$.
    Notice that $|W^{\pi}_{i,j}|=|W^{\pi'}_{i,4j-3}|=|W^{\pi'}_{i,4j-2}|=|W^{\pi'}_{i,4j-1}|=|W^{\pi'}_{i,4j}|$. 
    Define the weighted wiggle $WW_{i,j}$ as
    \[WW_{i,j}=\frac{\pi'(i)(j)+\pi'(i)(j+1)+\pi'(i+1)(j)+\pi'(i+1)(j+1))}{4}|W^{\pi'}_{i,j}|^p.\]
    We have for $i\in [n]$ and $j\in [\ell-1]$ that
    \[WW_{i,4j-3}+WW_{i,4j-2}+WW_{i,4j-1}+WW_{i,4j}=4M|W^\pi_{i,j}|.\]
    It follows that $\pi$ is a solution of \pwigglemin\ with value $x$ if and only if $\pi'$ is a solution of \wpwigglemin\ with value $4Mx$. Note that $F$ has to be balanced for this, as otherwise the wiggle values $|W^{\pi'}_{n,j}|$ for $j\in [4\ell]$ would only be weighted by $M/2$ and not~$M$.
\end{proof}
}

A trivial algorithm for the mentioned problems would need to go over all permutations, and thus need $\mathcal{O}(n!|I|)$ time when $|I|$ is the instance size. A simple improvement is to use dynamic programming by computing the optimal solution for each subset of time series (or elements in $S$), leading to an algorithm needing time $\mathcal{O}(2^{n}|I|^c)$ for some constant $c$. In the following sections, we want to find out whether polynomial time exact, or approximation algorithms for the problems are likely.

\section{Complexity of \texorpdfstring{\prefsumsumprob\ and \probname{(Weighted)-}\onewigglemin}{}}\label{section:prefsumsum}
In this section, we consider the problem \prefsumsumprob. We will show several complexity results, some of which will imply results for \onewigglemin\ and \wonewigglemin. \iflong We start with a proof of strong \NP-hardness.\fi
\subsection{Strong NP-hardness}
Most of this section is devoted to showing strong \NP-hardness of \prefsumsumprob. The proof idea is to reduce instances of a known \NP-hard problem to instances $S$ of \prefsumsumprob, such that $S$ has a solution value which is equal to a lower bound if and only if the instance which we reduced from is a yes instance. This lower bound is stated next.
\begin{lemma}\label{lemma:lowerbound}
    A lower bound for the objective function of \prefsumsumprob\ is $\sum_{s\in S}|s|/2$.
\end{lemma}
\begin{proof}
    Let $\pi$ be some solution of \prefsumsumprob.
    It is easy to see that for $i=1,\dots,n$,
    \begin{align*}
        |P^\pi_i|+|P^\pi_{i-1}|\ge |\pi(i)|.
    \end{align*}
    Thus,
    \begin{align*}
        2\sum_{i=0}^n |P^\pi_i|&\ge |P^\pi_0|+|P^\pi_n|+2\sum_{i=1}^{n-1} |P^\pi_i|\ge \sum_{i=1}^n |s_i|\qedhere
    \end{align*}
\end{proof} 
We call a solution $\pi$ achieving this bound \emph{bound-achieving}.
The next lemma states a sufficient condition that implies that the solution value exceeds the previously stated lower bound.
\begin{lemma}\label{lemma:abovelowerbound}
    Let $\pi$ be a solution of \prefsumsumprob\ and $\sum_{i=1}^ns_i=0$. 
    There exists $j\in [n]$ such that $|P^\pi_j|+|P^\pi_{j-1}|>|\pi(j)|$ if and only if $\sum_{i=1}^n|P^\pi_i|> \sum_{i=0}^n|s_i|/2$.
\end{lemma}
\begin{proof} The backwards direction is clear. For the forward direction, we have
    \begin{equation*}
        2\sum_{i=0}^n |P^\pi_i|=2\sum_{i=1}^{n-1}|P^\pi_i|>\sum_{i=1}^n |s_i|. \qedhere
    \end{equation*}
\end{proof}
Additionally, in our reduction, we will consider instances $S$ of \prefsumsumprob\ that have specific properties. These properties are described below.
\begin{enumerate}
\renewcommand{\labelenumi}{\textbf{\theenumi}}
        \renewcommand{\theenumi}{\mylipgray{(IP\arabic{enumi})}}
    \item $\sum_{i=1}^ns_i=0$.\label{enum:prop1}
    \item $n\equiv 0 \pmod{3}$, i.e., $n$ is a multiple of $3$.\label{enum:prop2}
    \item $S$ has exactly $\frac{2n}{3}$ positive elements and $\frac{n}{3}$ negative elements.\label{enum:prop3}
\end{enumerate}
Hence, $S$ also has no elements equal to zero.
We call an instance that satisfies \ref{enum:prop1}-\ref{enum:prop3} \emph{nice}, and observe the following properties that will be useful for the reduction.
\begin{lemma}\label{lemma:nobothpositive}
    Let $\pi$ be a bound-achieving solution of a nice instance $S$.
    At least one of $\pi(1), \pi(2)$ is negative and at least one of $\pi(n-1),\pi(n)$ is negative.
\end{lemma}
\begin{proof}
    If both $\pi(1), \pi(2)$ are positive, then $\pi(2)<|P^\pi_1|+|P^\pi_2|$. \cref{lemma:abovelowerbound} gives us a contradiction, as we assumed that $\pi$ is bound-achieving.

    If both $\pi(n-1),\pi(n)$ are positive, then we have, as $P^\pi_n=0$, that $P^\pi_{n-1}=-\pi(n)$ and $P^\pi_{n-2}=-(\pi(n)+\pi(n-1))$. It follows that $|P^\pi_{n-2}|+|P^\pi_{n-1}|>|\pi(n-1)|$, a contradiction.
\end{proof}

\begin{lemma}\label{lemma:2conspi0}
    Let $\pi$ be a bound-achieving solution of a nice instance $S$. If for $i\in [n-1]$, $\pi(i)$ and $\pi(i+1)$ are both positive or both negative, then $P^\pi_i=0$.
\end{lemma}
\begin{proof}
    If $\pi(i)$ and $\pi(i+1)$ are both positive, and $P^\pi_i\ne 0$ then we have two cases.
    \begin{itemize}
        \item If $P^\pi_i>0$ then $|P^\pi_i|+|P^\pi_{i+1}|>\pi(i+1)$, a contradiction.
        \item If $P^\pi_i<0$ then $|P^\pi_{i-1}|+|P^\pi_i|>\pi(i)$, a contraction.
    \end{itemize}
    When both $\pi(i)$ and $\pi(i+1)$, are negative, the proof works the same.
\end{proof}

The following is a direct consequence of the last lemma.
\begin{corollary}\label{lemma:threeconsecutive}
    Let $\pi$ be a bound-achieving solution of a nice instance $S$. There exists no $i\in [n-2]$ such that $\pi(i),\pi(i+1),\pi(i+2)$ are all negative or all positive. 
\end{corollary}
With these lemmas out of the way, we can show that $\pi$ has a unique structure with regard to the signs of its elements.
\begin{lemma}\label{lemma:signproperties}
    Let $\pi$ be a bound-achieving solution of a nice instance $S$. Then, for $i\in [n]$,
    \begin{itemize}
        \item $\pi(i)$ is positive if and only if $i\equiv 0\pmod{3}$ or $i\equiv 1\pmod{3}$,
        \item $\pi(i)$ is negative if and only if $i\equiv 2\pmod{3}$, and
        \item if $i\equiv 0\pmod{3}$, then $P^\pi_i=0$.
    \end{itemize}
\end{lemma}
\begin{proof}
    We show the claim by induction on $i$. For the base case, $i=1$, assume for a contradiction that $\pi(1)$ is negative. Together with  \cref{lemma:nobothpositive} and the pigeonhole principle, it follows that there must be three consecutive positive elements, a contradiction to \cref{lemma:threeconsecutive}.

    For the induction step $i>1$, we have the following cases.
    \begin{itemize}
        \item If $i\equiv 0\pmod{3}$ then $\pi(i)$ must be positive; otherwise, for the remaining $n-i$ following elements, we have $(n-i)/3-1$ negative elements and $2(n-i)/3+1$ positive elements. By the pigeonhole principle, there are three consecutive positive elements, a contradiction. Further $P^\pi_i=0$ must hold: For $i=n$ it holds trivially. Otherwise, if $P^\pi_i$ was not zero, then $\pi(i+1)$ must be negative by \cref{lemma:2conspi0}. For there to be no three consecutive positive elements in the remaining array, $\pi(n-1)$ and $\pi(n)$ must be positive, a contradiction.
        \item If $i\equiv 1\pmod{3}$ then $\pi(i)$ is positive by the same argumentation as in the base case.
        \item If $i\equiv 2\pmod{3}$ then $\pi(i)$ is clearly negative; otherwise, as $P^\pi_{i-1}$ is positive, $|P^\pi_i|+|P^\pi_{i+1}|>|\pi(i)|$ would hold; a contradiction (see \cref{lemma:abovelowerbound}). \qedhere
    \end{itemize}
\end{proof}

With this groundwork, we are able to show the following result. 
\begin{theorem}
    The decision variant of \prefsumsumprob\ is strongly \NP-complete.
\end{theorem}
\begin{proof}
    \NP-membership is obvious, thus we continue with \NP-hardness.
    We reduce from the strongly \NP-hard problem \textsc{Numerical 3-Dimensional Matching} \cite{DBLP:journals/scheduling/YuHL04}. The problem input consists of three multisets $X,Y$, and $Z$, each containing $k$ integers.
    Let $b=(\sum_{e\in X\cup Y\cup Z}e)/k$. The problem is to decide whether there exists a subset $M$ of $X\times Y\times Z$ such that
    \begin{enumerate}
        \item each element from $X$, $Y$, and $Z$ appears in exactly one triple, and
        \item $x+y+z=b$ holds for each triple $(x,y,z)\in M$.
    \end{enumerate}
    We can assume that each element from $X\cup Y\cup Z$ is positive (otherwise add to all elements $-\alpha+1$, where $\alpha=\min X\cup Y\cup Z$ and obtain an equivalent instance with regard to the solution). We also assume that each element from $X,Y$, and $Z$ is smaller than $b$.
    Now define the multisets $X'=\{x\mid x\in X\}$, $Y'=\{y+2b\mid y\in Y\}$, and $Z'=\{z-3b\mid z\in Z\}$. Let $n=|X'\cup Y'\cup Z'|$. Notice that $\sum_{e\in X'\cup Y'\cup Z'}e=0$, and $X',Y'$ and $Z'$ are pairwise disjoint.
    We now define the instance $S$ of \prefsumsumprob\ as $S=X'\cup Y'\cup Z'$. Notice that $S$ is a nice instance satisfying \ref{enum:prop1}-\ref{enum:prop3}.

    We claim that $(X,Y,Z)$ is a yes-instance of \textsc{Numerical 3-Dimensional Matching} if and only if there exists a solution $\pi$ of $S$ that is bound-achieving. We prove both directions.

    ``$\Rightarrow$'': Let $M$ be a subset of $X\times Y\times Z$ constituting a solution to the instance of \textsc{Numerical 3-Dimensional Matching}. We constructively define $\pi$ by going through the triples $(x,y,z)$ in $M$ in any order.
    Let $(x,y,z)$ be the $i$th triple.
    We define $\pi(3i-2)=x, \pi(3i-1)=z-3b$, and $\pi(3i)=y+2b$.
    By the definition of $X',Y'$, and $Z'$ such a permutation always exists.
    Now notice for $j=0,1,\dots, n$ that $P^\pi_j$ is
    \begin{itemize}
        \item $0$, if and only if $j\equiv 0\pmod{3}$,
        \item positive if and only if $j\equiv 1\pmod{3}$, and
        \item negative if and only if $j\equiv 2\pmod{3}$.
    \end{itemize}
    Thus, we have $|P^\pi_j|+|P^\pi_{j-1}|=|\pi(j)|$ for all $j\in [n]$. It follows that $\sum_{j=1}^n |P^\pi_j|=\sum_{s\in S} |s|/2$.

    ``$\Leftarrow$'': Assume that $\pi$ is bound-achieving. We construct $M$, starting with $M=\emptyset$.
    Because $S$ is nice, we have that \cref{lemma:signproperties} holds for $\pi$. For each $i\in [k]$ we have the following. Let $u:=\pi(i-2),v:=\pi(i-1)$, and $w:=\pi(i)$. Because of \cref{lemma:signproperties} we have that $u+v+w=0$, $v$ is negative, and $u$ and $w$ are positive. Hence, $v\in Z'$ holds. Now we claim that exactly one of $u$ and $w$ is from $X'$ and exactly one is from $Y'$. Indeed, if both $u$ and $w$ were from $X'$, then $u+w<2b<|v|$ and $u+v+w=0$ is impossible. Similarly, if both $u$ and $w$ were from $Y'$, then $u+w>4b>|v|$ and $u+v+w=0$ is impossible.
    Thus, w.l.o.g.\ assume that $u\in X'$ and $w\in Y'$. Add to $M$ the triple $(u,w-2b,v+3b)$.
    It is easy to verify that after adding this triple for each $i\in [k]$, $M$ verifies that $(X,Y,Z)$ is a yes-instance of \textsc{Numerical 3-Dimensional Matching}.
\end{proof}

\cref{lemma:prefsumsumtowigglemin} and \cref{lemma:wigglemintoweightedwigglemin} imply the following corollary. 
\begin{corollary}
    The decision variants of \onewigglemin\ and \wonewigglemin\ are strongly \NP-complete, even for a constant number of time points.
\end{corollary}

\subsection{Restricted instances of \texorpdfstring{\prefsumsumprob}{}}
We further ask for which types of instances \prefsumsumprob\ becomes tractable. In particular, we want to know whether restricting the number of negative (or positive) elements of $S$ has an effect on the computational complexity of \prefsumsumprob. Indeed, if $S$ only contains positive or only negative elements, then an optimal ordering $\pi$ will simply order the elements of $S$ by their absolute value.  But what if $S$ contains exactly one positive or exactly one negative element? We obtain the following peculiar results, which state that the sum of $S$ has an effect on the complexity of \prefsumsumprob\ in this case.
\begin{restatable}[\appsymb]{theorem}{theoremprefsumsumpoly}\label{theorem:prefsumsumpoly}
    \prefsumsumprob\ can be solved in time $\mathcal{O}(n\log n)$ if $S$ contains exactly one negative (positive) element, and $\sum_{i=1}^ns_i=0$.
\end{restatable}
\begin{restatable}[\appsymb]{theorem}{theooremprefsumsumweaknphard}\label{theoorem:prefsumsumweaknphard}
    The decision variant of \prefsumsumprob\ is \NP-complete, even if $S$ contains exactly one negative (positive) element.
\end{restatable}
We dedicate the rest of the section to show both results, starting with \cref{theorem:prefsumsumpoly}.

\ifshort
\begin{proof}[Proof sketch of \cref{theorem:prefsumsumpoly}]
    We show the result for $S$ containing a single negative element $x$, the result for a single positive element is shown equivalently.
    Now consider some permutation $\pi$ of $S$, and let $k=\text{pos}_\pi(x)$. We observe that the objective value can be split into absolute prefix sums before and after position $k$ as follows.
    \begin{align*}
        \sum_{i=1}^n|\sum_{j=1}^{i}\pi(j)|&=\sum_{i=1}^{k-1}\sum_{j=1}^i\pi(j)+\sum_{i=1}^{n-k}\sum_{j=1}^i\pi(n-j+1)\\
        &=\sum_{i=1}^{k-1}\pi(i)\cdot (k-i)+\sum_{i=1}^{n-k}\pi(n-i+1)\cdot (n-k-i+1).
    \end{align*}
     Thus, we obtain a weighted sum of the positive elements. The coefficient of an element depends only on its position in relation to the position of the negative element.
    This leads to a simple algorithm, which is also given as pseudocode in the long version of this paper: The permutation is chosen such that the negative element is in the middle. The smallest positive elements are placed at positions $1$ and $n$, the next smallest at positions $2$ and $n-1$, etc.
\end{proof}
\fi
\appendixproof{theorem:prefsumsumpoly}{
\ifshort\theoremprefsumsumpoly*\fi
\begin{proof}
\begin{algorithm}[htb]
\caption{Exact algorithm for \prefsumsumprob\ when $S$ contains exactly one negative element and $\sum_{s\in S}s=0$.}\label{alg:easyalg}
    \KwIn{A multiset $S$ of integers, which contains exactly one negative element.}
    \KwOut{A permutation $\pi$ of $S$.}
    $\pi(\lfloor \frac{n}{2} \rfloor)\gets \min S$\;
    $B\gets S\setminus \min S$\;
    $i\gets 0$\;
    \ForEach{$s\in B$ in non-decreasing order}{
        \uIf{$i \equiv 1\pmod{2}$}{
            $\pi(\frac{i+1}{2})\gets s$\;
        }
        \Else{
            $\pi(n-\frac{i}{2})\gets s$\;
        }
        $i\gets i+1$\;        
    }
    \Return $\pi$;
\end{algorithm}
    We show the result for $S$ containing a single negative element $x$, the result for a single positive element is shown equivalently.
    Now consider some permutation $\pi$ of $S$, and let $k=\text{pos}_\pi(x)$. We observe that the objective value can be split into absolute prefix sums before and after position $k$ as follows.
    \begin{align*}
        \sum_{i=1}^n|\sum_{j=1}^{i}\pi(j)|&=\sum_{i=1}^{k-1}\sum_{j=1}^i\pi(j)+\sum_{i=1}^{n-k}\sum_{j=1}^i\pi(n-j+1)\\
        &=\sum_{i=1}^{k-1}\pi(i)\cdot (k-i)+\sum_{i=1}^{n-k}\pi(n-i+1)\cdot (n-k-i+1).
    \end{align*}
    Thus, we obtain a weighted sum of the positive elements. The coefficient of an element depends only on its position in relation to the position of the negative element. Moreover, at most two elements can have coefficient one, at most two can have two, and so on. Thus, it is beneficial to have the two largest positive elements at positions $k-1$ and $k+1$, the next largest at position $k-2$ and $k+2$, etc.
    This leads to the simple \cref{alg:easyalg}. The permutation $\pi$ is chosen such that the negative element is in the middle. Then the smallest positive elements are placed at positions $1$ and $n$, the next smallest at positions $2$ and $n-1$, etc. This is optimal because in this way the largest positive elements are only counted towards the objective once, the next-largest two elements twice, and so on. Putting the negative element in the middle is optimal because it results in the smallest coefficients for the positive elements. The runtime is upper-bounded by sorting the numbers in $S$.
\end{proof}
}

To show \cref{theoorem:prefsumsumweaknphard}, we reduce from the following problem.

\begin{problem}[\oneintwopartition]
    Given is a set $X$ containing $m$ pairs of positive integers. No integer appears twice in the union of all the integers. The question is: Are there two sets $X_1,X_2$ of size $m$ such that both sets contain exactly one element from each pair in $X$, no element is contained in both sets, and $\sum_{x\in X_1} x=\sum_{x\in X_2} x$?
\end{problem}
Clearly, \oneintwopartition\ is \NP-hard by a straightforward reduction from \probname{Partition} \cite{DBLP:books/fm/GareyJ79}: Let $Y=\{y_1,\dots,y_n\}$ be an instance of \probname{Partition}, asking for a subset $Y'\subseteq Y$ with $\sum_{y\in Y'} y'=\frac{1}{2}\sum_{y\in Y}y$. We can assume all integers in $Y$ to be positive.
Let $b=\sum_{y\in Y}y$. The reduced instance consists of the pairs $(i\cdot b+y_i, i\cdot b)$ for $i=1,\dots,n$. The correctness of this reduction is immediate. We are ready to prove \cref{theoorem:prefsumsumweaknphard}.

\ifshort
\begin{proof}[Proof sketch of \cref{theoorem:prefsumsumweaknphard}]
    Note that the proof is rather technical; to understand it in full detail, the interested reader should refer to the full description in full version of the paper.
     We reduce from \oneintwopartition. Consider an instance $X=\{p_1,p_2,\dots,p_m\}$ of \oneintwopartition, where $m>1$. Define $M=\sum_{(x,y)\in X}\frac{x+y}{2}$. We define the following three sets.
    \begin{align*}
        S_1&=\{m(x+iM)\mid i\in [m],p_i=(x,y)\}\\
        S_2&=\{m(y+iM)\mid i\in [m],p_i=(x,y)\}\\
        S_3&=\{imM+1\mid i\in [m]\}
    \end{align*}
    We continue with further definitions.
    \begin{align*}
        M_1&=\frac{\sum_{s\in S_1}s+\sum_{s\in S_2}s}{2}=mM+mM\sum_{i=1}^mi\\
        M_2&=\sum_{s\in S_3}s=m+mM\sum_{i=1}^mi\\
        S&=S_1\cup S_2\cup S_3\cup \{-(M_1+M_2)\}
    \end{align*}
    For each $(x,y)\in X$ there are unique values $x',y'\in S$ such that for some $i$, $x'=m(x+iM)$ and $y'=m(y+iM)$. We say that $x'$ is the \emph{correspondent} of $x$ and $y'$ is the \emph{correspondent} of $y$, and vice versa.
    We observe that  $S\setminus \{-(M_1+M_2)\}$ can be partitioned uniquely into sets $t_1,t_2,\dots, t_m$ such that for each $t_i$ we have that
    \begin{itemize}
        \item $t_i$ contains exactly three elements,
        \item $t_i$ contains one element from $S_1$, one from $S_2$, and one from $S_3$, 
        \item the correspondent of $\{x\}=S_1\cap t_i$ and the correspondent of $\{y\}=S_2\cap t_i$ form a pair in $X$, and
        \item each integer from $t_i$ is in the interval $[imM, (i+1)mM)$.\label{property:proofsketch2}
    \end{itemize}
    
    We set $S$ as the instance of \prefsumsumprob\ and claim that $X$ is a yes-instance of \oneintwopartition\  if and only if there exists a permutation $\pi$ of $S$ such that
    \begin{equation*}
    \sum_{i=1}^n|\sum_{j=1}^{i}\pi(j)|\le \sum_{i=1}^m (m-i+1)\sum_{x\in t_i}x.
    \end{equation*}

    We sketch the more interesting direction of the proof, assuming that $\pi$ is a solution achieving this objective, showing how to construct $X_1$ and $X_2$.
    The idea now is that $\pi$ will order the elements as shown in \cref{fig:prefdistribution}. Namely, the prefix  sum $P_{2m+1}$ will be zero, and the last $m$ elements will be ordered increasingly. Further, $\pi(2m+1+i)$ will be from set $t_i$ for each $i\in[n]$. This further implies that the last $m$ elements sum to $M_1$, which is a multiple of $m$. Together, this means that the last $m$ elements cannot contain any elements from $S_3$. Now let $X_1$ be the set of correspondents of the last $m$ elements in $\pi$, and let $X_2$ be the remaining elements from $X$. The above properties imply that $\sum_{x\in X_1}x=\sum_{x\in X_2}x$.
\end{proof}
\fi
\begin{figure}
    \centering
    \includegraphics{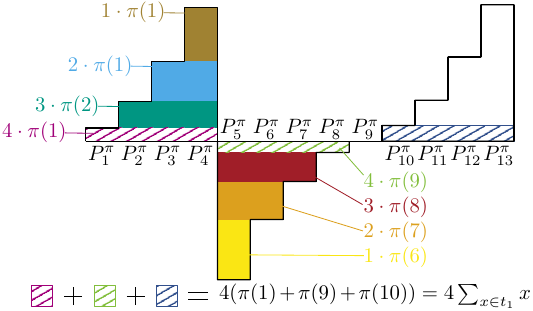}
    \caption{Schematization of the distribution of prefix sums for an optimal solution $\pi$ of $S$ in the proof of \cref{theoorem:prefsumsumweaknphard}, where $m=4$. The height corresponds to the prefix sum value, while values below the horizontal middle line are negative.
    The figure illustrates how the objective value (the area of the union of all enclosed regions) can be decomposed into a weighted sum of the positive elements. Further, the three smallest elements at positions $1,9$, and $10$ contribute to the sum with coefficient $4$, resulting in the hatched regions. In an optimal solution, these elements correspond to the three smallest values, constituting $t_1$.}
    \label{fig:prefdistribution}
\end{figure}
\appendixproof{theoorem:prefsumsumweaknphard}{
    \ifshort\theooremprefsumsumweaknphard*\fi
    \begin{proof}[Proof of \cref{theoorem:prefsumsumweaknphard}]
    We reduce from \oneintwopartition. Consider an instance $X=\{p_1,p_2,\dots,p_m\}$ of \oneintwopartition, where $m>1$. Define $M=\sum_{(x,y)\in X}\frac{x+y}{2}$, which we assume to be a multiple of $2$. We may further assume that all integers from $X$ are smaller than $M$; otherwise $X$ is a no-instance. We define the following three sets.
    \begin{align*}
        S_1&=\{m(x+iM)\mid i\in [m],p_i=(x,y)\}\\
        S_2&=\{m(y+iM)\mid i\in [m],p_i=(x,y)\}\\
        S_3&=\{imM+1\mid i\in [m]\}
    \end{align*}
    We continue with further definitions.
    \begin{align*}
        M_1&=\frac{\sum_{s\in S_1}s+\sum_{s\in S_2}s}{2}=mM+mM\sum_{i=1}^mi\\
        M_2&=\sum_{s\in S_3}s=m+mM\sum_{i=1}^mi\\
        S&=S_1\cup S_2\cup S_3\cup \{-(M_1+M_2)\}
    \end{align*}
    We observe the following crucial properties of the sets and constants defined above.
    \begin{enumerate} %
        \renewcommand{\labelenumi}{\textbf{\theenumi}}
        \renewcommand{\theenumi}{\mylipgray{(P\arabic{enumi})}}
        \item The sets $S_1,S_2$ and $S_3$ are pairwise disjoint.
        \item Let $A\subseteq S$ with $A\cap S_3\ne \emptyset$ and $\sum_{a\in A}a\equiv 0\pmod{m}$. Then $S_3\subseteq A$.\label{property:propS2}
        \item For each $(x,y)\in X$ there are unique values $x',y'\in S$ such that for some $i$, $x'=m(x+iM)$ and $y'=m(y+iM)$. We say that $x'$ is the \emph{correspondent} of $x$ and $y'$ is the \emph{correspondent} of $y$, and vice versa.
        \item $S\setminus \{-(M_1+M_2)\}$ can be partitioned uniquely into sets $t_1,t_2,\dots, t_m$ such that for each $t_i$ we have that
        \begin{itemize}
            \item $t_i$ contains exactly three elements,
            \item $t_i$ contains one element from $S_1$, one from $S_2$, and one from $S_3$, 
            \item the correspondent of $\{x\}=S_1\cap t_i$ and the correspondent of $\{y\}=S_2\cap t_i$ form a pair in $X$, and
            \item each integer from $t_i$ is in the interval $[imM, (i+1)mM)$.
        \end{itemize}
        Thus, notice that $x\in t_i$ and $y\in t_j$ with $i<j$ implies $x<y$.\label{property:propS3}
    \end{enumerate}
    We set $S$ as the instance of \prefsumsumprob\ and claim that $X$ is a yes-instance of \oneintwopartition\  if and only if there exists a permutation $\pi$ of $S$ such that
    \begin{equation}
    \sum_{i=1}^n|\sum_{j=1}^{i}\pi(j)|\le \sum_{i=1}^m (m-i+1)\sum_{x\in t_i}x.\label{eq:sumprefixprf2},
    \end{equation}
    where $n=|S|=3m+1$.
    Before showing both directions in detail, let us discuss the right side of \eqref{eq:sumprefixprf2}; we observe that the sum constitutes a weighted sum of only the positive elements in $S$. The three largest elements in $S$, corresponding to $t_m$, have coefficient one, the elements from $t_{m-1}$ have coefficient $2$, and so on (refer to \cref{fig:prefdistribution}). In fact, due to the structure of the instance, this is a lower bound for the objective value, which will be discussed later in more detail.
    We are ready to prove both directions.

    ``$\Rightarrow$'': Let $X_1$, $X_2$ be sets  witnessing that $X$ is a yes-instance of \oneintwopartition.
    Let $\pi_{S_3}$ be the list of numbers in $S_3$ in ascending order.
    We define $\pi$ as follows. 
    \begin{itemize}
        \item For $i=1,\dots,m$, $\pi(i)=\pi_{S_3}(i)$.
        \item For $i=m+1$, $\pi(i)=-(M_1+M_2)$.
        \item For $i=m+2,\dots, 2m+1$, $\pi(i)=a$, where $a$ is defined as follows. Let $j=i-m-1$. Consider $t_{m-j+1}$. Let $\{x\}=S_1\cap t_{m-j+1}$ and $\{y\}=S_2\cap t_{m-j+1}$. Let $x'$ be the correspondent of $x$ and $y'$ be the correspondent of $y$. We know that either $x'$ is in $X_1$ or $y'$ is in $X_1$. If $x'$ is in $X_1$, we define $a$ as $x$, otherwise as $y$.
        \item For $i=2m+2,\dots, 3m+1$, $\pi(i)=a$, where $a$ is defined as follows. Let $j=i-2m-1$. Consider $t_{j}$. Let $\{x\}=S_1\cap t_{j}$ and $\{y\}=S_2\cap t_{j}$. Let $x'$ be the correspondent of $x$ and $y'$ be the correspondent of $y$. We know that either $x'$ is in $X_2$ or $y'$ is in $X_2$. If $x'$ is in $X_2$, we define $a$ as $x$, otherwise as $y$.
    \end{itemize}
    As $\sum_{x\in X_1} x=\sum_{x\in X_2} x$, we have that $\sum_{i=m+1}^{2m+1}\pi(i)=M_2$. Due to $\sum_{i=1}^m\pi(i)=M_1$ and $\pi(m+1)=-(M_1+M_2)$,  it follows that $\sum_{i=1}^{2m+1}\pi(i)=0$. With this, it is easy to see (also refer to \cref{fig:prefdistribution}) that 
    \[\sum_{i=1}^n|\sum_{j=1}^{i}\pi(j)|=\sum_{i=1}^m (m-i+1)\sum_{x\in t_i}x.\]

    ``$\Leftarrow$'': Assume that $\pi$ is such that \eqref{eq:sumprefixprf2} holds. We first want to show that the prefix sums $P_i^\pi$ look as in \cref{fig:prefdistribution}.
    Let $k$ be such that $\pi(k)=-(M_1+M_2)$.
    First, we claim that $P^\pi_k$ cannot be non-negative, otherwise we would have the following.
    \begin{align}
        \sum_{i=1}^n|P^\pi_i|&\ge \sum_{i=1}^{k-1} (k-i)\pi(i)+\sum_{i=k+1}^n (n-i+1)\pi(i)\label{eq:coeffsubopt}\\
        &>\sum_{i=1}^m (m-i+1)\sum_{x\in t_i}x \label{eq:coeffopt}
    \end{align}
    The first inequality is obtained by decomposing the distribution of prefix sums into a weighted sum of the positive elements in $S$, similarly as in \cref{fig:prefdistribution}.
    For the second inequality, consider the coefficients of on the right side of \eqref{eq:coeffsubopt}. Each coefficient $1,2,3,\dots$ appears at most twice. As each positive element of $S$ occurs in the sum, some elements have coefficients larger than $m$.
    On the other hand, in \eqref{eq:coeffopt} each coefficient $1,2,3,\dots,m$ appears exactly three times. Further, coefficients are distributed optimally, i.e., coefficient $1$ is applied to the largest values in $t_m$, coefficient $2$ to the values in $t_{m-1}$, and so on. Thus, \eqref{eq:coeffopt} is clearly smaller.    
    Thus, we can observe that the sequence of $P^\pi_i$'s is non-negative for $i=0,\dots k-1$, negative for $i=k,\dots,\ell$ for some $\ell$, and lastly non-negative again for $i=\ell+1,\dots,n$. 
    
   We again consider the decomposition of $\sum_{i=1}^n|P^\pi_i|$ into a weighted sum of elements from $S^+=S\setminus \{-(M_1+M_2)\}$. More importantly, due to the aforementioned properties of the $P^\pi_i$'s, at most three elements from $S^+$ have coefficient one, at most three have coefficient $2$, and so on. The remaining elements all have larger coefficients $3,4,5,\dots$. The only exception is the element $\pi(\ell+1)$, as it might be that $P^\pi_\ell$ is negative and $P^\pi_{\ell+1}$ is positive. 
    For this case, the contribution of $\pi(\ell+1)$ to $\sum_{i=1}^n|P^\pi_i|$ can be expressed as $|P^\pi_\ell|\cdot (\ell-k+1)+|P^\pi_{\ell+1}|\cdot (n-\ell)$, with $|P^\pi_\ell|+|P^\pi_{\ell+1}|=\pi(\ell+1)$. Notice that in that case, some element from $S^+$, $|P^\pi_\ell|$, or $|P^\pi_{\ell+1}|$ have coefficient $m+1$.
    The obvious best distribution of coefficients (achieving the smallest sum of absolute prefix sums) is such that the three largest elements from $S^+$ get coefficient one, the next three smaller elements get coefficient $2$, and so on, with the largest coefficient being $m$.
    Notice that this is only possible when 
    \begin{enumerate}
        \renewcommand{\labelenumi}{\textbf{\theenumi}}
        \renewcommand{\theenumi}{\mylipgray{(P\alph{enumi})}}
        \item the first $m$ elements of $\pi$ are positive in increasing order,\label{propsredpi1}
        \item the next element is $-(M_1+M_2)$,\label{propsredpi2}
        \item the next $m$ elements being positive in decreasing order, with $P^\pi_{2m+1}=0$, \label{propsredpi3}
        \item the last $m$ elements being positive in increasing order, and\label{propsredpi4}
        \item due to the optimal choice of coefficients, this further implies that $\pi(k-1),\pi(k+1),\pi(n)$ are the largest three elements from $S^+$, $\pi(k-2),\pi(k+2),\pi(n-1)$ are the next-largest, and so on.\label{propsredpi5}
    \end{enumerate}   
    A solution with these properties is illustrated in \cref{fig:prefdistribution}.
    Now, we observe that this optimal choice of coefficients is exactly what is represented with the sum $\sum_{i=1}^m (m-i+1)\sum_{x\in t_i}x$. It follows, that in fact $\sum_{i=1}^m (m-i+1)\sum_{x\in t_i}x$ is a lower bound for $\sum_{i=1}^n|P^\pi_i|$, and \ref{propsredpi1}-\ref{propsredpi5} hold for $\pi$. Due to \ref{propsredpi5}, we have that $\{\pi(k-i), \pi(k+i),\pi(n-i+1)\}$ constitutes the triple $t_{m-i+1}$ for $i=1,\dots, m$. Let $A$ constitute the first $m$ positive elements of $\pi$, $B$ the next $m$ positive elements, and $C$ the last $m$ positive elements. As $P^\pi_{2m+1}=0$ we know that $\sum_{c\in C} c=M_1$, which is a multiple of $m$. Thus, $C\cap S_3=\emptyset$ by \ref{property:propS2}, and by the two facts $|C|=|S_3|=m$ and $M_1\ne M_2$. Hence, $C$ only consists of elements from $S_1\cup S_2$. Moreover, $C$, contains exactly one element from each triple $t_i$ due to \ref{property:propS3} and \ref{propsredpi5}.
    Now let $C'$ be the set of correspondents of elements in $C$. We have $\sum_{c\in C'} c=M$, $|C|=m$, and $C$ contains exactly one element from each pair in $X$. Thus, by defining $X_1=C'$ and $X_2$ as the remaining integers from pairs in $X$, we have shown that $X$ is a yes-instance of \oneintwopartition.
\end{proof}
}
Note that a pseudo-polynomial algorithm for \prefsumsumprob\ cannot be ruled out for the case of one negative (positive) element, as we did not show strong \NP-hardness.

\section{\texorpdfstring{\pwigglemin}{} and \texorpdfstring{\wpwigglemin}{}}\label{section:pwigglemin}
In this section, we show hardness result for \pwigglemin\ and \wpwigglemin\ for arbitrary $p$. We present a reduction which has several complexity implications in \cref{sec:minlinarrred}.
\cref{sec:approx} contains results with regard to approximation lower bounds.

\subsection{A Reduction Implying NP-hardness}\label{sec:minlinarrred}
The reduction is from the well-known problem \minlinarrlong\ \cite{DBLP:books/fm/GareyJ79} to \pwigglemin. The problem is defined as follows.
\begin{problem}[\minlinarrlong\ (\minlinarr)]
    Given an undirected graph $G=(V,E)$ and an integer $k$, does there exist a permutation $\pi$ of $V$ such that \[\sum_{\{u,v\}\in E}|\text{pos}_\pi(u)-\text{pos}_\pi(v)|\le k?\]
\end{problem}
Consider a constant positive integer $p$. The reduction will work for any such $p$, and goes as follows. Let $(G,k)$ be an instance of \minlinarr\ with $V=\{v_1,\dots, v_n\}$ and $E=\{e_1,\dots,e_m\}$.
We construct a set $F=\{f_1,\dots,f_n\}$ of $(m+1)$-time series.
We define the time series inductively based on the time point, starting with the first time point: let $f_i(1)=m$ for all $i=1,\dots n$. For time point $j>1$ consider the edge $e_{j-1}=\{v_a,v_b\}$ with $a<b$. We define $f_i(j)$ as
\begin{equation}
    f_i(j)=\begin{cases}
        f_i(j-1)+1&\qquad\text{if }i=a,\\
        f_i(j-1)-1&\qquad\text{if }i=b,\\
        f_i(j-1)&\qquad\text{otherwise.}
    \end{cases}
\end{equation}
Now consider an arbitrary permutation $\pi^V$ of $V$. Let $\pi^F$ be a permutation of $F$ that is defined such that $\text{pos}_{\pi^V}(v_i)=\text{pos}_{\pi^F}(f_i)$ for all $i\in [n]$.
We have the following.
\begin{restatable}[\appsymb]{lemma}{lemmalinarreduction}\label{lemma:linarrreduction}
    It holds that
    \begin{equation}
        \sum_{i=1}^n\sum_{j=1}^m|W_{i,j}^{\pi^F}|^p=\sum_{\{u,v\}\in E}|\text{pos}_\pi(u)-\text{pos}_\pi(v)|.
    \end{equation}
\end{restatable}
\appendixproof{lemma:linarrreduction}{
    \ifshort\lemmalinarreduction*\fi
    \begin{proof}
    The proof is due to a simple calculation.
    \begin{align*}
        \sum_{\{u_a,u_b\}\in E}|\text{pos}_\pi(u_a)-\text{pos}_\pi(u_b)|&=\sum_{\{u_a,u_b\}=:e_i\in E,a<b}(\sum_{i=1}^{a-1}0^p+\sum_{i=a}^{b-1}1^p+\sum_{i=b-1}^n0^p)\\
        &=\sum_{j=1}^m(\sum_{i=1}^{a-1}|W^{\pi^F}_{i,j}|^p+\sum_{i=a}^{b-1}|W^{\pi^F}_{i,j}|^p+\sum_{i=b}^n|W^{\pi^F}_{i,j}|^p)\\
        &=\sum_{i=1}^n\sum_{j=1}^m |W^{\pi^F}_{i,j}|^p \qedhere
    \end{align*}
\end{proof}
}

Thus, a solution $\pi^V$ with solution value $k$ w.r.t.\ \minlinarr\ corresponds to a solution $\pi^F$ with solution value $k$ w.r.t.\ \pwigglemin. Hence, hardness and inapproximability results of \minlinarr\ directly carry over to \pwigglemin. Further, the instance $F$ is balanced, so results also carry over to \wpwigglemin\ by \cref{lemma:wigglemintoweightedwigglemin}. We obtain the following by the \NP-hardness of \minlinarr~\cite{DBLP:books/fm/GareyJ79}. 
\begin{theorem}
    Let $p\ge 1$ be an arbitrary integer. The decision variants of \pwigglemin\ and \wpwigglemin\ are strongly \NP-complete.
\end{theorem}

\subsection{Approximation Lower Bounds}\label{sec:approx}
In this section, we consider the complexity of approximating \pwigglemin\ and \wpwigglemin. Firstly, the reduction from the previous section implies two hardness results due to hardness results for \minlinarr\ shown Ambühl et al.~\cite{DBLP:journals/siamcomp/AmbuhlMS11} and Raghavendra et al.~\cite{DBLP:conf/coco/RaghavendraST12}.
\begin{theorem}
    Let $p\ge 1$ be an arbitrary integer. Let $\epsilon>0$ be an arbitrarily small constant. If there is a PTAS for \pwigglemin\  or for \wpwigglemin, then there is a (probabilistic) algorithm that decides whether a given SAT instance of size $n$ is satisfiable in time $2^{n^\epsilon}$.
\end{theorem}
Note that it is widely believed that such an algorithm for \probname{SAT} is unlikely, thus it is unlikely that there exists a PTAS for both problems.
\begin{theorem}
    Let $p\ge 1$ be an arbitrary integer. Under the Small-Set Expansion Hypothesis \cite{DBLP:conf/stoc/RaghavendraS10}, there is no constant-factor approximation for \pwigglemin\ and \wpwigglemin.
\end{theorem}
The Small-Set Expansion Hypothesis is a hypothesis that implies the Unique-Games Conjecture~\cite{DBLP:conf/stoc/Khot02a,DBLP:conf/stoc/RaghavendraS10}. It was introduced by Raghavendra and Steurer in 2010 \cite{DBLP:conf/stoc/RaghavendraS10} and has since received much attention, as this conjecture would imply some hardness and inapproximability results, including the non-existence of a constant-factor approximation for treewidth \cite{DBLP:journals/jair/WuAPL14}.

\subparagraph*{A known greedy heuristic.}
Next, we want to look at a known greedy heuristic called \mylipgray{BestFirst} for \wonewigglemin\ that has been applied in works on stacked area charts \cite{strungemathiesenAestheticsOrderingStacked2021,dibartolomeoThereMoreStreamgraphs2016}. Given an instance $F=\{f_1,\dots,f_n\}$ of \wonewigglemin, the heuristic iteratively builds the order $\pi$, starting from the empty order. At step $i=1,\dots,n$ of the heuristic, the partial ordering of length $i-1$ is extended by appending a not yet chosen time series to the end that has the smallest increase in the objective value. Such heuristics equivalently exist for \onewigglemin\ and \prefsumsumprob, we also call them \mylipgray{BestFirst}.
\begin{restatable}[\appsymb]{theorem}{thmbestfirst}\label{thm:bestfirst}
    The \mylipgray{BestFirst} heuristic has approximation factor at least $\Omega(\sqrt[3]{n})$ for \wonewigglemin, \onewigglemin, and \prefsumsumprob.
\end{restatable}
\appendixproof{thm:bestfirst}{
    \ifshort\thmbestfirst* \fi
    \begin{proof}
    We first present instances, that show the result for \prefsumsumprob. In the end of the proof we show how, from these instances, instances for \onewigglemin\ and \wonewigglemin\ can be constructed.
    Let $m$ be some positive multiple of $2$.
    Consider the multiset $T$ consisting of
    \begin{itemize}
        \item $2m$ of the number $1$,
        \item two of the number $m^2-m$, and
        \item one of the number $-2m^2$.
    \end{itemize}
    Let $S$ consist of $m^2$ copies of $T$.
    We notice that $n=|S|=3m^2+2m^3$.
    Let $\pi$ be the solution that is obtained when applying \mylipgray{BestFirst} to the instance $S$. The heuristic will always append the number that achieves the smallest absolute prefix sum in the next step. Thus, the first $(m/2)(m^2+m+2)$ elements of $\pi$ consist of $m/2$ copies of the sequence
    \[(\overbrace{1,1,\dots,1}^{m^2\text{ ones}},-2m^2,m^2-m,\overbrace{1,1,\dots,1}^{m\text{ ones}}),\]
    whose sum is zero. We have that
    \begin{align*}
        \sum_{i=1}^n|P^\pi_i|&\ge \sum_{i=1}^{m/2(m^2+m+2)}|P^\pi_i|\\
        &=(m/2)\left(\sum_{i=1}^{m^2}i+m^2+m+\sum_{i=1}^mi\right)\\
        &=(m/2)\left(\frac{m^2(m^2+1)}{2}+m^2+m+\frac{m(m+1)}{2}\right)\\
        &=\Omega(m^5)
    \end{align*}
    Now consider the solution $\pi^*$ to the instance $S$ which consists of $m^2$ copies of
    \[(\overbrace{1,1,\dots,1}^{m\text{ ones}},m^2-m,-2m^2,m^2-m,\overbrace{1,1,\dots,1}^{m\text{ ones}}).\]
    We have
    \begin{align*}
        \sum_{i=1}^n|P^{\pi^*}_i|&=m^2\left(\sum_{i=1}^mi+m^2+m^2+m+\sum_{i=1}^mi\right)\\
        &=m^2\left(2\frac{m(m+1)}{2}+2m^2\right)\\        
        &=\mathcal{O}(m^4).
    \end{align*}
    So for such instances, the approximation ratio of the \mylipgray{BestFist} for \prefsumsumprob is at least $\Omega(m)=\Omega(\sqrt[3]{n})$.

    Now, given such an instance $S$ of \prefsumsumprob, we can construct an instance $F$ for \wpwigglemin\ (and \pwigglemin) as follows. For each $s\in S$, we add to $F$ a $2$-time series $f$ such that $f(1)=m^2-s/2$ and $f(2)=m^2+s/2$. We observe that $f(2)-f(1)=s$ and that $(f(1)+f(2))/2$ is $m^2$, the latter observation is independent of $f$. Further, $F$ is balanced. As $f(2)+f(1)=2m^2$ independent of $f$, the instance $F$ behaves equivalently for \wonewigglemin\ and \onewigglemin\ with regard to \mylipgray{BestFirst}. Because for each $s\in S$ there exists an $f\in F$ with $f(2)-f(1)=s$ and vice versa, it also behaves equivalently to the instance $S$ of \prefsumsumprob, in the sense that it also implies the same lower bound on the approximation factor of \mylipgray{BestFirst}.
\end{proof}
}

\section{MILP and Experiments for \texorpdfstring{\wonewigglemin}{}}\label{section:experimental}
Here, we present a mixed-integer linear program (MILP) for \wonewigglemin\ and compare it to the state-of-the art heuristic from Mathiesen and Schulz~\cite{strungemathiesenAestheticsOrderingStacked2021}. This heuristic is different from the \bestfirst\ heuristic from the last section and will be explained in more detail later. The aim of this comparison is to determine how an exact approach for \wpwigglemin\ scales with respect to input size and to estimate how the heuristic of \cite{strungemathiesenAestheticsOrderingStacked2021} compares with an exact algorithm with respect to solution quality. Note that the heuristic we compare with is strictly better than \bestfirst\ as was demonstrated in an experimental evaluation by Mathisen and Schulz~\cite{strungemathiesenAestheticsOrderingStacked2021}.
\ifshort Due to space constraints, the description of the MILP can be found in in the full version.\fi
\toappendix{
\subsection{A Mixed-Integer Linear Program for \texorpdfstring{\wonewigglemin}{}}\label{section:milp}
In this section, we present an MILP that solves \wonewigglemin\  optimally. Let $F=\{f_1,\dots,f_n\}$ be an instance of \wonewigglemin\ on $\ell$ time points. For the description of the MILP, let $F_{\ne}^2=\{(f_i,f_j)\mid (f_i,f_j)\in F\times F, i\ne j\}$ and let $F_{\ne}^3=\{(f_i,f_j,f_k)\mid (f_i,f_j,f_k)\in F\times F\times F, i\ne j, j\ne k, i\ne k\}$.
The MILP makes use of the following variables.
\begin{itemize}
    \item A binary variable $x_{i,j}$ for each $(f_i,f_j)\in F_{\ne}^2$ which is $1$ if and only if $f_i$ comes before $f_j$ in~$\pi$.
    \item A real variable $y_{i,j}$ for each $i\in [n]$ and $j\in [\ell]$ which will correspond to the sum of time series at time point $j$ up to time series $f_i$.
    \item A real variable $z^{\text{abv}}_{i,j}$ for each $i\in [n]$ and $j\in [\ell-1]$. This variable will correspond to the absolute wiggle between time points $j$ and $j+1$ for the upper border of time series $f_i$.
    \item A real variable $z^{\text{bel}}_{i,j}$ for each $i\in [n]$ and $j\in [\ell-1]$. This variable will correspond to the absolute wiggle between time points $j$ and $j+1$ for the lower border of time series $f_i$. 
\end{itemize}
The MILP is given below.
\begin{align}
    \min\qquad& \sum_{i=1}^n\sum_{j=1}^{\ell-1} \frac{f_i(j)+f_i(j+1)}{4}(z^{\text{abv}}_{i,j}+z^{\text{bel}}_{i,j})\tag{OBJ} \label{eq-obj}\\
    x_{i,j}=1-x_{j,i}&\qquad (f_i,f_j)\in F_{\ne}^2\tag{SYM}\label{eq-sym}\\
    x_{i,j}+x_{j,k}+x_{k,i}\le 2&\qquad(f_i,f_j,f_k)\in F_{\ne}^3\tag{TRANS}\label{eq-trans}\\
    y_{i,j}=f_i(j)+\sum_{k\in [n]\setminus\{i\}}x_{k,i}f_k(j)&\qquad i\in[n],j\in [\ell]\tag{FIXY}\label{eq-fixy}\\
    z^{\text{abv}}_{i,j}\ge y_{i,j}-y_{i,j+1}&\qquad i\in [n],j\in [\ell-1]\tag{Z1}\label{eq-Z1}\\
    z^{\text{abv}}_{i,j}\ge y_{i,j+1}-y_{i,j}&\qquad i\in [n],j\in [\ell-1]\tag{Z2}\label{eq-Z2}\\
    z^{\text{bel}}_{i,j}\ge (y_{i,j}-f_i(j))-(y_{i,j+1}-f_i(j+1))&\qquad i\in [n],j\in [\ell-1]\tag{Z3}\label{eq-Z3}\\
    z^{\text{bel}}_{i,j}\ge (y_{i,j+1}-f_i(j+1))-(y_{i,j}-f_i(j))&\qquad i\in [n],j\in [\ell-1]\tag{Z4}\label{eq-Z4}\\
    x_{i,j}\in \{0,1\}&\qquad (f_i,f_j)\in F_{\ne}^2\tag{BIN}\label{eq-bin}
\end{align}
The idea is to split the contributions to the objective of $f_i(j)$ and $f_i(j+1)$, $i\in [n]$, $j\in [\ell-1]$, into two parts -- the weighted wiggle $\frac{f_i(j)+f_i(j+1)}{4}z^{\text{abv}}_{i,j}$ above the time series $f_i$ between time points $j$ and $j+1$ and the weighted wiggle $\frac{f_i(j)+f_i(j+1)}{4}z^{\text{bel}}$ below. The objective \eqref{eq-obj} simply sums over all weighted wiggle values. The constraints \eqref{eq-sym} and \eqref{eq-trans} ensure that the $x$-variables encode a total order. \eqref{eq-fixy} sets the $y$-variables. \eqref{eq-Z1}-\eqref{eq-Z4} compute absolute values, \eqref{eq-Z1} and \eqref{eq-Z2} together with the objective set $z_{i,j}^{\text{abv}}=|y_{i,j}-y_{i,j+1}|$, \eqref{eq-Z3} and \eqref{eq-Z4} together with the objective value set $z^{\text{bel}}_{i,j}=|(y_{i,j}-f_i(j))-(y_{i,j+1}-f_i(j+1))|$. The number of constraints is cubic in $n$.}

\subsection{Experimental Setup}
\subparagraph*{Tested algorithms.}
We compared two algorithms for \wonewigglemin. The first, \MILP, is an implementation of the MILP\iflong from \cref{section:milp}\fi. The second, \Upwards, is the best state-of-the art algorithm from \cite{strungemathiesenAestheticsOrderingStacked2021}. It iteratively performs \emph{moves} that improve the objective value. In a single move, a time series is removed from the current ordering and reinserted into the position that is best with regard to the objective value.

\subparagraph*{Instances.} 
We obtained instances from real-world data, and these instances were also used in \cite{strungemathiesenAestheticsOrderingStacked2021}. For example, the eight instances from \cite{strungemathiesenAestheticsOrderingStacked2021} include unemployment rates of countries for a span of months, Covid values of countries for a span of days, and movie revenues for a span of weeks. These eight instances are very large, with 33--1,000 time series and 33--243 time points. We cannot expect that such instances can be solved to optimality by \MILP, thus we sampled from these instances sub-instances as follows. For each instance $I$ and each $k=10,15,\dots,60$, we picked uniformly at random $k$ time series from $I$ (if $I$ contains at least $k$ time series). These $k$ time series then constitute an instance. We did this five times for each combination of $k$ and $I$, resulting in 465 instances overall. 
\subparagraph*{Hardware and setup.}
Both algorithms were implemented in Python 3.12.4. We used the original implementation of \Upwards\ from \cite{strungemathiesenAestheticsOrderingStacked2021} with the same parameters. \MILP\ was implemented using the Python interface of Gurobi, using Gurobi v11.0.2. Due to large floating point values in the computations, the \texttt{NumericFocus} parameter of Gurobi was set to~3. The \texttt{MIPGap} parameter was set to~0, in order to find optimal solutions. Without this setup we ran into numeric issues in preliminary experiments.
The experiments were run on a cluster using Intel Xeon E5-2640 v4, 2.40GHz 10-core processors, running Ubuntu 18.04.6 LTS. To simulate an end-user machine, the memory limit was set to 8GB and each algorithm was executed on a single thread. The time limit was set to one hour. 
\subsection{Results}
\begin{figure}
    \begin{subfigure}[b]{0.4\textwidth}
        \includegraphics{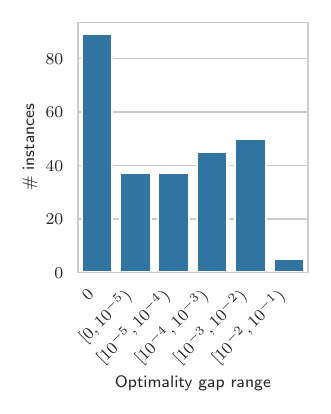}
        \caption{Number of instances by the optimality gap achieved by \Upwards.}
        \label{fig:gap}
    \end{subfigure}
    \hfill
    \begin{subfigure}[b]{0.59\textwidth}
        \includegraphics{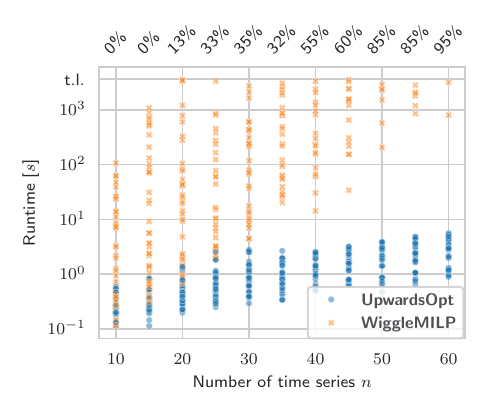}
        \caption{Scatter plot of runtimes by the number of time series $n$. The numbers at the top represent the percentage of instances with the respective value $n$ that timed out for \MILP.}
        \label{fig:runtime}
    \end{subfigure}
    \caption{Results of the experiments.}
\end{figure}
In this section, we answer two questions: first, how does \Upwards\ compare to \MILP\ with respect to the solution quality; second, how scalable are both approaches with respect to the size of an instance. For the first question, we only considered 263 out of 465 instances where \MILP\ produced the optimal solution -- it neither timed out nor ran into the memory limit. For each instance where \MILP\ produced the optimal solution, we compute the \emph{optimality gap} $g$ that is defined as $g=\text{sol}_{\text{\Upwards}}/\text{sol}_{\text{\MILP}}-1$, where $\text{sol}_x$ is the solution value of the algorithm $x$. Thus, a value of zero corresponds to \Upwards\ computing the optimal solution and larger values represent ``how far away'' \Upwards\ is from the optimum. \cref{fig:gap} shows the distribution of optimality gaps. We observe that \Upwards\ could solve 89 out of the 263 instances optimally, while the remaining optimality gaps are relatively low. This confirms that that \Upwards\ is a good heuristic for \wpwigglemin\ for real-world instances.

Next, in \cref{fig:runtime} we present a scatter plot of runtimes. The $y$-axis shows runtime in a logscale. The $x$-axis corresponds to the instance size in terms of the number of time series. Each point in the plot corresponds to an instance-algorithm combination. The values at the top denote the percentage of instances with respective value $n$ that timed out for \MILP\ (the memory limit was never reached). There were no timeouts for \Upwards. We observe that \Upwards\ is relatively fast, while \MILP\ runs for much longer and times out already for some instances with $20$ time series.

Generally, we conclude that the state-of-the art heuristic \Upwards\ is well-suited for \wonewigglemin\ for real-world instances, even though the problem is computationally very hard in the theoretical sense.

\section{Conclusion and Open Problems}
We have investigated variants of wiggle minimization in stacked area charts from the theoretical side and showed computational lower bounds. Not only is it \NP-hard to minimize wiggle, but it is also unlikely that approximations with good approximation guarantees exist. Nonetheless, we could show in an experimental evaluation that an existing heuristic is very good at minimizing weighted wiggle for real-world instances. Due to this being the first theoretical work on minimizing wiggle in stacked area charts, there remain some open problems.
\begin{itemize}
    \item Are there constant-factor approximations for \prefsumsumprob?
    \item \sloppy{Are there stronger inapproximability results for \pwigglemin\ and \wpwigglemin, i.e., under the assumption that $\P\ne \NP$.}
    \item We think that improvements to \MILP\ are possible, and it might be interesting to investigate other exact  wiggle minimization approaches.
    \item Finally, our results might be extended to wiggle minimization in streamgraphs. It seems likely that the hardness results in this paper carry over to streamgraphs, though, we were not able to show that yet.
\end{itemize}
\bibliography{literature}

\end{document}